\documentclass[10pt]{IEEEtran}
\ifCLASSINFOpdf
\else
\fi
\usepackage{subfig}
\usepackage{array}
\usepackage{subfig}
\usepackage{amsmath}
\usepackage[nolist,nohyperlinks]{acronym}
\usepackage{graphicx,wrapfig,lipsum}
\usepackage{rotating}
\usepackage{amsmath, amssymb, amsthm}
\usepackage{stmaryrd}
\usepackage{tikz}
\usetikzlibrary{shapes,arrows,matrix,shapes,positioning,chains,decorations.pathreplacing,angles,quotes}
\usepackage{verbatim}
\usepackage{url}
\usepackage[utf8]{inputenc}
\usepackage[english]{babel}

\newtheorem{corollary}{Corollary}[]

\newtheorem{lemma}[]{Lemma}
\usepackage{multicol}
\usepackage{xr-hyper}
\usepackage{algorithm}
\usepackage{algpseudocode}
\usepackage{scalerel}
\usepackage{multicol}
\usepackage{setspace}

\usepackage[noadjust]{cite}
\usepackage{bbm}
\usepackage[normalem]{ulem}
\usepackage{color}
\usepackage{dsfont}
\newtheorem{rmk}{Remark}[]
\usetikzlibrary{automata}
\usetikzlibrary{arrows}
\allowdisplaybreaks
\begin{document}
	%
	\title{Positioning Data-Rate Trade-off in mm-Wave Small Cells and Service Differentiation for 5G Networks}
	\author{$^*$Gourab Ghatak$^{1,2}$, $^*$Remun Koirala$^{1,3,4}$, Antonio De Domenico$^1$, Benoit Denis$^1$, Davide Dardari$^4$, and Bernard Uguen$^3$
     \\ \small{ $^1$CEA, LETI, MINATEC, F-38054 Grenoble,
France; $^2$LTCI, Telecom ParisTech, Universit\'e Paris Saclay, France; $^3$University of Rennes 1-IETR (CNRS UMR 6164), Rennes, France; $^4$DEI, University of Bologna, Cesena (FC), Italy}
\vspace*{-0.0cm}}
		\maketitle
	\begin{abstract}
We analyze a millimeter wave network, deployed along the streets of a city, in terms of positioning and downlink data-rate performance, respectively. First, we present a transmission scheme where the base stations provide jointly positioning and data-communication functionalities. Accordingly, we study the trade-off between the localization and the data rate performance based on theoretical bounds. Then, we obtain an upper bound on the probability of beam misalignment based on the derived localization error bound. Finally, we prescribe the network operator a scheme to select the beamwidth and the power splitting factor between the localization and communication functions to address different quality of service requirements, while limiting cellular outage. 
	\end{abstract}
    \vspace*{-0.0cm}
\section{Introduction}
{\let\thefootnote\relax\footnote{{$^*$These authors have equal contribution in the work.\\
The research leading to these results has received funding from European Union under grant n. 723247 and supported by the Institute for Information \& communications Technology Promotion (IITP) grant funded by the Korea government (MSIP) (No.B0115-16-0001, 5GCHAMPION).}}}
To address the multi-fold increase in the demand for data rates, exploitation of higher frequency spectrum in the millimeter wave (mm-wave) range is gaining popularity~\cite{rappaport2013millimeter}. However, mm-wave communication is characterized by high path loss and sensitivity to blockages. To solve these problems, beam-forming techniques are utilized with the help of highly directional antennas, which result in new issues in terms of coverage and initial access~\cite{li2016initial}. Moreover, beam-alignment errors between the base stations (BSs) and the user equipments (UEs) degrade the communication performance.
One solution to this problem consists of enabling UEs to simultaneously receive signals in the mm-wave and in the sub-6GHz band, and to use the latter to support the initial access on the mm-wave band~\cite{ghatak2017coverage}. 
Another approach exploits positioning algorithms to support the UE cell discovery and access to mm-wave BSs. On the one hand, with fine-tuned positioning, the beam-alignment procedure is quickened, and beamforming and user tracking are improved~\cite{garcia2016location}. On the other hand, improved mm-wave beam-forming can be used for more accurate localization and orientation of nodes~\cite{Destino2017on}. 

In addition to the high speed data rates, the fifth generation (5G) cellular networks anticipate an explosion of new services, characterized by heterogeneous requirements. 
We investigate a mm-wave network deployed for supporting positioning and broadband functionalities simultaneously, e.g., in vehicle-to-infrastructure communication. Specifically, we study the trade-off between positioning efficiency and downlink data rates and accordingly, we prescribe the operator an algorithm to tune the mm-wave BS transmit power so as to meet specific quality of service (QoS) requirements of different services.
 \vspace*{-0.3cm}
\subsection{Related Work}
 \vspace*{-0.2cm}
In the context of sub-6GHz systems, Jeong et al.~\cite{jeong2015beamforming} have studied a distributed antenna system providing both data communication and positioning functionalities. The authors assumed that the UEs know the positions of the BSs and attempt to estimate their own positions based on the received signals. Lemic et al.~\cite{lemic2016localization} have shown that localization using mm-wave frequencies is efficient in terms of accuracy, even in the presence of a limited number of anchor nodes. In fact, mm-wave beam-forming allows for accurate localization and orientation of UEs with respect to the BSs~\cite{Destino2017on}. Garcia et al.~\cite{garcia2016location} have studied a location-aided initial access strategy for mm-wave networks, in which the information of UE locations enables to speed up the channel estimation and beam-forming procedures. Destino et al.~\cite{Destino2017on} have studied the trade-off between communication rate and positioning quality in a single user mm-wave link. Similarly Koirala et al.~\cite{koirala17localization} have studied the beamforming optimization and spectral power allocation based on theoretical localization bounds.

The downlink communication performance in random wireless networks is typically characterized by signal to interference and noise ratio (SINR) coverage probability and rate coverage probability, using stochastic geometry~\cite{bai2015coverage}. For this, the positions of the BSs are modeled using homogeneous Poisson point process (PPP)~\cite{elshaer2016downlink} or using repulsive point processes~\cite{chun2015modeling}. Recently, Ghatak et.al.~\cite{8254855} investigated a more realistic scenario, where mm-wave BSs are deployed along the roads of a city. We use this model in this paper, and accordingly we study a one dimensional setting where the BSs and the served users are assumed to be on the same street. 

Specifically, leveraging on the tools of stochastic geometry, we present an average characterization of the localization and communication performance of this network, by exploiting the a-priori knowledge about the distribution of the distances of the users from the BSs. We analyze the positioning and data communication trade-off, and provide the operator with a power control scheme designed to satisfy distinct QoS requirements of the positioning and the communication functions.
 \vspace*{-0.2cm}
\subsection{Contributions and Organization}
The main contributions of this paper are:
\begin{itemize}
\item We characterize a noise-limited mm-wave system designed to support positioning and broadband services simultaneously by partitioning the BS transmit power. First, we obtain the Cramer-Rao lower bound (CRLB) for the estimation of the distance of a typical UE from its serving BS. Subsequently, we obtain the signal to noise ratio (SNR) and rate coverage probability of the typical user, as a function of the power splitting factor.
\item Leveraging on the derived CRLB for the estimation of the distance, we obtain an upper bound on the probability of beam-misalignment. Based on this, we compute the minimum antenna beamwidth that limits the beam-misalignment.
\item Finally, we analyze the trade-off between the positioning and the data rate performance of the typical user. Accordingly, we prescribe the operator with a scheme to select the proper power splitting factor to support different QoS requirements. Specifically, we study our mm-wave system under different operating beamwidths, and analyze the distribution of the total transmit power for maximizing either the positioning efficiency or the UE data-rate.
\end{itemize}
The rest of the paper is organized as follows. In Section~\ref{sec:SM}, we introduce our system model and outline the performance objectives. In Section~\ref{sec:MR}, we derive our main results on the positioning error, the rate coverage, and the misalignment error. We provide some numerical results in Section~\ref{sec:NRD}, and accordingly present our power partitioning scheme. Finally, the paper concludes in Section~\ref{sec:Con}.
\vspace*{-0.2cm}
\section{System Model}
\label{sec:SM}
We consider an urban scenario, with multi-storied buildings that result in a dense blocking environment. In this scenario we analyze a mm-wave network consisting of BSs deployed along the streets of the city.
\vspace*{-0.3cm}
\subsection{Network Geometry}
The positions of the BSs in each street are modeled as points of a one-dimensional Poisson point process (PPP) $\phi$, with intensity $\lambda$ [m$^{-1}$]. Each BS is assumed to be of known height $h_B$ and equipped with directional antennas with beamwidth $\theta$. Let the corresponding product of the directivity gains of the transmitting and receiving antennas be $G_0$. The transmit power of the BSs is assumed to be $P$. Without loss of generality we perform our analysis from the perspective of a typical user located at origin, which associates with the BS that provides the highest downlink power. 
Accordingly, the distribution of the distance $d$ of the typical user from the serving BS is given by~\cite{chiu2013stochastic}:
\begin{align}
f_d(x) = 2\lambda\exp(-2\lambda x)
\label{eq:dist}
\end{align}
Furthermore, we assume that the network is equipped with efficient interference management capabilities (e.g., spatio-temporal frequency reuse), so that the performance of the users is noise-limited\footnote{Although the assumption of the network being noise-limited simplifies the analysis, Singh et al.~\cite{5733382} have shown the validity of this assumption in outdoor mm-Wave mesh networks. In a future work, we will extend the analysis by considering interfering BSs.}.
\vspace*{-0.2cm}
\subsection{Path-loss}
\vspace*{-0.1cm}
Due to the low local scattering, we consider a Nakagami fading for mm-wave communications~\cite{7593259} with parameter $n_0$ and variance equal to 1. Furthermore, we assume a path loss model where the power at the origin received from a BS located at a distance $d$ is given by $P_r = K\cdot P \cdot g \cdot G_0 \cdot (d^2 + h_B^2)^{\frac{-\alpha}{2}}$, where $K$ is the path loss coefficient, $g$ represents the fast-fading, and $\alpha$ is the path loss exponent. Thus, the average SNR can be written as $\frac{K\cdot P \cdot G_0 \cdot (d^2 + h_B^2)^{\frac{-\alpha}{2}}}{N_0 \cdot B}$. $N_0$ and $B$ are the noise power density and the operating bandwidth, respectively.
\vspace*{-0.5cm}
\subsection{Transmission Policy}
We assume a communication scheme where the transmit power of the BSs is divided into two parts: one associated with positioning and the other allotted for data communication. The power allocated for localization determines the number of control symbols used for this function, whereas the remaining power is utilized for control and data symbols of the communication phase. We acknowledge that it is possible to utilize the native communication signal for positioning services. However, we use dedicated waveforms designed for better localization performance (e.g., see \cite{damman16optimizing} for a discussion on localization specific waveforms). Hence, splitting of the transmit power becomes necessary to characterize and optimize the operating trade-off between communication and localization functionalities.
Accordingly, if the total transmit power is $P$, and $\beta$ is the fraction of power used for data services, the corresponding transmit power for localization is $P_L = (1 - \beta)P$. Consequently, the transmit power for data service is $P_D = \beta P$. Let the SNR for the distance estimation and the data communications phases be represented by $SNR_1$ and $SNR_2$, respectively.
\vspace*{-0.2cm}
\section{Positioning Error, Data Rate Coverage and Misalignment Error}
\label{sec:MR}
In this section, we first characterize the minimum variance of the error in the estimation of the distance of the typical user from the serving BS. Then, we derive the SNR coverage and the rate coverage probabilities. 
\vspace*{-0.2cm}
\subsection{Distance Estimation Analysis}
To simplify our analysis, we only consider the effect of the distance on the power of the received signal (for instance, we consider Received Signal Strength Indicator (RSSI) based ranging algorithms), and ignore the effect of the distance on the phase~\cite{wu2012fila}. Accordingly, the received signal is:
\begin{equation}
	y(t) = \frac{\sqrt{KG_0P_L}}{(h_B^2+d^2)^{\frac{\alpha}{4}}} x\left(t\right) + n(t),
\end{equation}
where $n(t)$ is a zero mean additive white Gaussian noise resulting in estimation errors.
\begin{lemma}
The expected value of the Fisher information for the estimation of the distance ($d$) is calculated as:
\begin{align}
J_D = \frac{KG_0P_L  2 \lambda \bar{f^2}}{\sigma_N^2}  \int_1^\infty \frac{e^{-2 \lambda x}}{(h_B^2+x^2)^{\frac{\alpha}{2}}} d x,
\end{align}
where $\bar{f^2} = 1.25\pi^2B^2$.
Furthermore, the prior information is:
$J_p = \log\left(2\lambda \right) - 1$.
\label{lem:FIM}
\end{lemma}
\begin{proof}
The Fisher information for a given $d$ is~\cite{van2004detection}:
\begin{equation}
	J_d = \frac{KG_0P_L}{(h_B^2+ d^2)^{\frac{\alpha}{2}} \sigma_N^2} \bar{f^2},
\end{equation}
where $\bar{f^2}= \frac{\int_{-\infty}^{\infty} (2 \pi f)^2 |X(f)|^2 df}{\int_{-\infty}^{\infty} |X(f)|^2 df}$ is the effective bandwidth of the signal. In our case, we assume that the signal has a flat spectrum~\cite{Destino2017on}, and accordingly, we have $\bar{f^2} = 1.25\pi^2B^2$. 
Now using the distribution of $d$ from~\eqref{eq:dist}, the expectation of the Fisher information is calculated as:
 \begin{eqnarray}
 	J_D = \mathbb{E}_d \left[ J_d \right] = \frac{KG_0P_L  2 \lambda \bar{f^2}}{\sigma_N^2}  \int_1^\infty \frac{e^{-2 \lambda x}}{(h_B^2+x^2)^{\frac{\alpha}{2}}} d x.  
 \end{eqnarray}
 Finally, the prior information can be calculated as:
\begin{align}
J_p &= \mathbb{E} \left[\log(f_d(x))\right] = \int_0^{\infty}   \log\left(f_d(x)\right)f_d(x) d x \nonumber \\
& = \int_0^{\infty} \log\left(2 \lambda  \exp\left(-2  \lambda x\right)\right) 2 \lambda \exp\left(-2 \lambda  x\right) dx \nonumber \\
& = \log\left(2\lambda \right) - 1\nonumber
\end{align}
This completes the proof.
\end{proof}
\begin{corollary}
For the special case of path loss exponent $\alpha = 2$, $J_D$ evaluates to~\eqref{eq:J_tau},
where $Ei$ is the exponential integral~\cite{pagurova1961tables}.
\begin{figure*}
\small
\begin{equation}
	J_D =  \frac{KG_0P_L  2 \lambda \bar{f^2}}{\sigma_N^2} \frac{i ( e^{-i 2 \lambda h} Ei(i 2 \lambda h) -e^{i 2 \lambda h} Ei(-i 2 \lambda h ) )}{2h} + 2\lambda \log\left(2\lambda \right) - 1 
    \label{eq:J_tau}
\end{equation}
\hrule
\end{figure*}
\end{corollary}
\vspace*{-0.4cm}
Finally, the Bayesian information can be obtained as $J_B = J_D + J_P$. Consequently, the Bayesian CRLB (BCRLB) and Jeffrey's prior corresponding to the Bayesian information are calculated as $\frac{1}{J_B}$ and $\sqrt{J_B}$, respectively.
\begin{rmk}
Intuitively, higher the Jeffrey's prior (or lower the BCRLB) is, better the estimation efficiency will be. From \eqref{eq:J_tau}, we see that a higher Jeffrey's prior is facilitated by a larger value of $P_L$, i.e., a smaller $\beta$. 
\end{rmk}
\vspace*{-0.4cm}
\subsection{Coverage and Rate Analysis}
Based on the path-loss model of Section II-B, the SNR for the communication phase is:
\begin{align}
SNR_2  = \frac{P_D K g G_0}{\sigma_N^2}  ({d^2 + h_B^2})^{-\frac{\alpha}{2}}. \nonumber 
\end{align}
Accordingly, let us define the SNR coverage probability of the typical user at a threshold $\gamma$, as the probability that the SNR is greater than $\gamma$. It represents the fraction of the users under coverage in the network.
\begin{lemma}
The SNR coverage probability at a threshold of $\gamma$ is calculated as~\eqref{eq:CovP}.
\begin{figure*}
\small
\begin{align}
&\mathcal{P}_C(\gamma) = \sum_{n = 1}^{n_0}\left(-1\right)^{n+1} \binom {n_0}n2\lambda \exp\left(2\lambda - \frac{h_B^2 n\gamma \sigma_N^2}{P_DKG_0}\right) \left[\frac{\sqrt{\pi}}{2}\left(\sqrt{\frac{P_DKG_0}{n\gamma \sigma_N^2}} - \frac{P_DKG_0}{n\gamma \sigma_N^2}\text{erf}\left(\frac{h_B^2 n\gamma \sigma_N^2}{P_DKG_0}\right)\right)\right]
\label{eq:CovP}
\end{align}
\hrule
\end{figure*}
\end{lemma}
\begin{proof}
The SNR coverage probability is computed as follows:
\begin{align}
&\mathbb{P}\left(SNR_2 \geq \gamma \right) = \mathbb{P}\left(\frac{P_D g K G_0}{\sigma_N^2}  (\sqrt{d^2 + h_B^2})^{-\alpha}  \geq \gamma \right)\nonumber \\
& = \mathbb{P}\left(g \geq \frac{\gamma \sigma_N^2}{P_DKG_0 ({x^2 + h_B^2})^{-\frac{\alpha}{2}}}\right) \nonumber \\
& = \sum_{n = 1}^{n_0}\left(-1\right)^{n+1} \binom {n_0}n\mathbb{E}\left[\exp\left(-\frac{n\gamma \sigma_N^2}{P_DKG_0 ({x^2 + h_B^2})^{-\frac{\alpha}{2}}}\right)\right] \nonumber \\
& = \sum_{n = 1}^{n_0}\left(-1\right)^{n+1} \binom {n_0}n2\lambda \int_0^{\infty}\exp\left(-\frac{n\gamma \sigma_N^2 ({x^2 + h_B^2})^{\frac{\alpha}{2}}}{P_DKG_0}\right) \cdot \nonumber \\
& \hspace*{4cm}\exp\left(-2\lambda x\right) dx \nonumber
\end{align}
Evaluating this integral completes the proof.
\end{proof}
Similar to the SNR coverage probability, the rate coverage probability at a threshold $r_0$ is defined as the probability that the downlink data rate of the typical user is greater than $r_0$. 
\begin{corollary}
The rate coverage probability can be computed as:
\begin{align}
\mathcal{P}_R(r_0) = \mathbb{P}\left(R \geq r_0\right) &= \mathbb{P} \left(SNR_2 \geq 2^{\frac{r_0}{B}} -1\right) \nonumber\nonumber \\ & =  P_C\left( 2^{\frac{r_0}{B}} -1\right)
\end{align}
\end{corollary}
\vspace*{-0.7cm}
\subsection{Beam Misalignment Error}
A BS with an antenna beamwith $\theta$, serving a user located at distance $d$, covers a region of length $D_0$ on the ground (see Figure 1). Using simple trigonometric calculations, we have:
\begin{align}
D_0 = \frac{2\tan\left(\frac{\theta}{2}\right)\left[1 + \frac{d^2}{h_B^2}\right]}{1 - \frac{d^2}{h_B^2}\tan^2\frac{\theta}{2}}. \nonumber
\end{align}
Once the localization procedure and the corresponding exchange of user-BS control signals is performed, beam-misalignment can occur in the absence of dynamic beam-alignment on both sides of the radio link. Assuming that the user's antenna is always oriented towards the BS, or equivalently, in case the user is operating with an omni-directional antenna, beam-misalignment will occur in case the distance of the user on the ground is more than $\frac{D_0}{2}$ from the estimated position. 

Let us assume that the estimation error for the UE localization is symmetric about its mean. Consequently, we bound the probability of the beam-misalignment as follows:
\begin{lemma}
The probability of beam-misalignment for a user located at a distance $d$ from the serving BS is bounded as $\frac{\text{BCRLB}}{D_0}$.
\begin{proof}
\begin{align}
\vspace*{-0.4cm}
\mathcal{P}_{MA}(d) = \mathbb{P}\left(|d - \hat{d}| \geq \frac{D_0}{2}\right) \stackrel{(a)}{\leq} \frac{2\sigma^2}{D_0} \stackrel{(b)}{=} \frac{2 \cdot \text{BCRLB}}{D_0},
\end{align}
where $\hat{d}$ is the estimated distance of the user. Here (a) follows from Markov's inequality assuming $\sigma^2$ as the variance of the positioning error. The step (b) occurs for an minimum-variance unbiased estimator (MVUE).
\vspace*{-0.3cm}
\end{proof}
\label{lem:missalignment}
\end{lemma}
\begin{corollary}
The mean misalignment error is then bounded by taking the expectation over d, i.e., $\bar{\mathcal{P}}_{MA} = \mathbb{E}_d\left[\mathcal{P}_{MA}(d)\right] \leq \mathbb{E}_d\left[\frac{2 \cdot \text{BCRLB}}{D_0}\right]$.
\end{corollary}
In the next section, we prescribe guidelines for an operator to choose an operating beamwidth for limiting this error.
\vspace*{-0.2cm}
\section{Numerical Results and Discussion}
\label{sec:NRD}
In this section, we present some numerical results based on the analytical framework presented in this paper. First, we show how the SNR coverage probability changes with the power splitting factor ($\beta$). Subsequently, we study the trade-off between localization and data rate as a function of $\beta$. Then, with the help of two examples, we describe our power partitioning scheme. In the following analysis, we assume $G_0 = 10$ dB and $n_0 = 3$.
\vspace*{-0.3cm}
\subsection{SINR Coverage Probability}
In Figure~\ref{fig:SINR} we plot the SNR coverage probability with respect to $\beta$ at a threshold of $\gamma = -10$ dB. As $\beta$ increases, the SINR coverage probability increases due to more power allocated to the data transmission phase. This provides a guideline to select a minimum operating $\beta$ for a given deployment density, such that the outage is limited. As an example, to limit a service outage below 20$\%$, with a BS deployment of 1 km$^{-1}$ and a power budget of $P = 25$ dBm, the minimum $\beta$ is 0.15, whereas with a power budget of $P = 20$ dBm, the minimum $\beta$ is 0.5.

More interestingly, this analysis provides the operator dimensioning rules in terms of the deployment density of the BSs for a given power budget. For example, in order to support services with an outage tolerance of 10$\%$, with a power budget of 20 dBm, a deployment density of 1 km$^{-1}$ does not suffice, and the operator must necessarily deploy more BSs.
\begin{figure}
\centering
\includegraphics[width = 8 cm, height = 4cm]{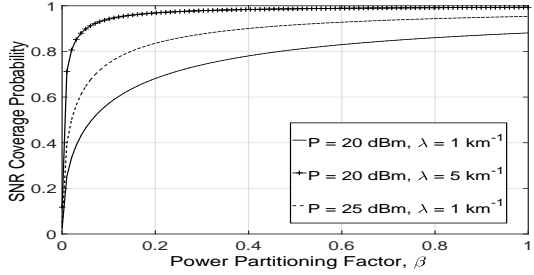}
\caption{SNR coverage probabilities for a threshold of $\gamma = -10$ dB vs the fractional power split for different $\lambda$.}
\label{fig:SINR}
\vspace*{-0.5cm}
\end{figure}
\vspace*{-0.2cm}
\subsection{Misalignment Error}
In Figure~\ref{fig:miss} we plot the mean beam-misalignment bound with respect to the beamwidth of the transmit antenna of the BSs. As expected, the larger the beamwidth and the higher the SNR, the lower the misalignment. For example, for a tolerable misalignment of 0.02$\%$ with SNR = -15 dB and $\lambda = 5$ km$^{-1}$, the minimum antenna beamwidth should be 8 degrees.
\begin{figure}
\centering
\includegraphics[width = 8 cm, height = 4cm]{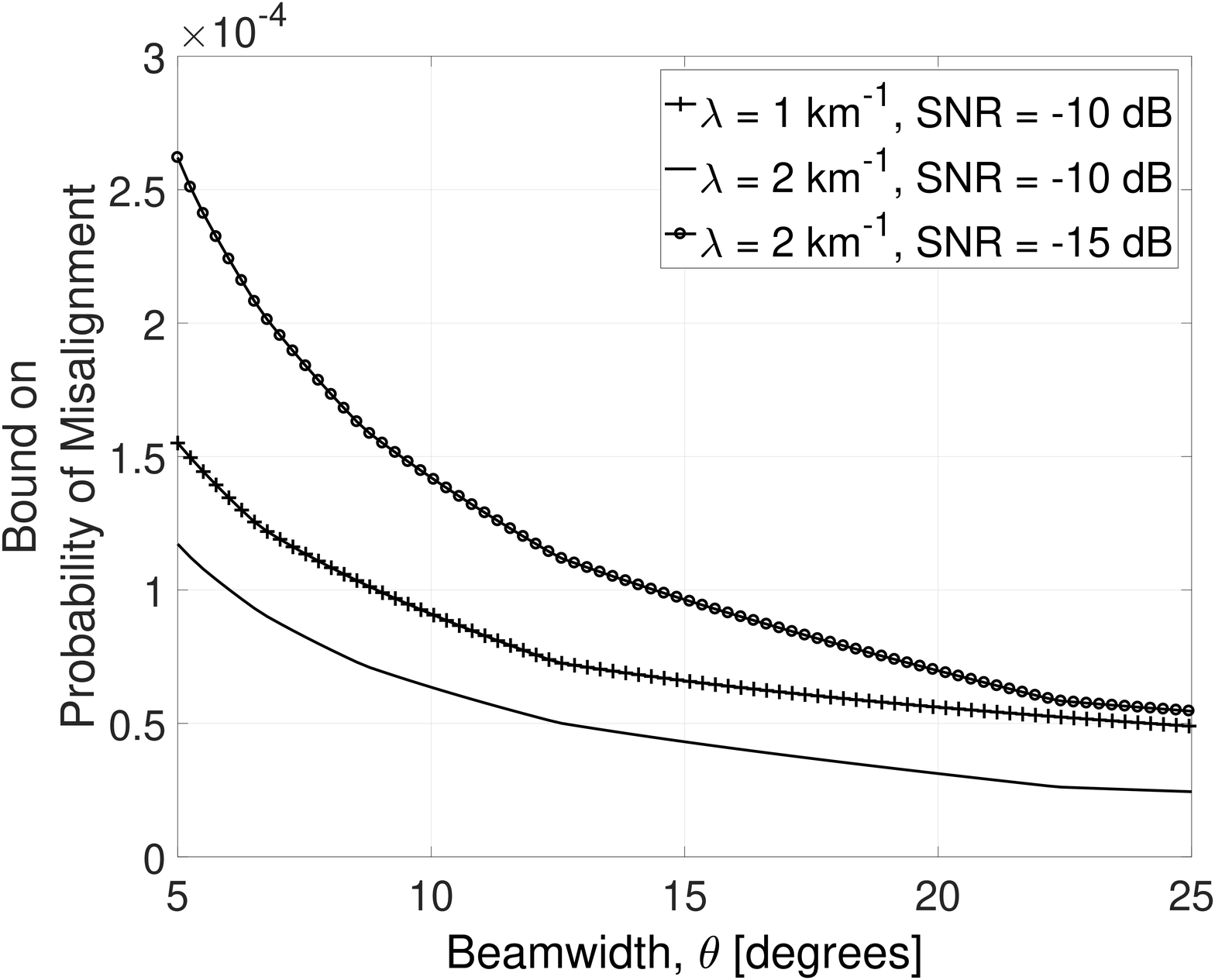}
\caption{Beam misalignment error with respect to beamwidth of the transmit antenna.}
\label{fig:miss}
\vspace*{-0.5cm}
\end{figure}
\vspace*{-0.5cm}
\subsection{Distance Estimation-Data Rate Trade-off}
In Figure~\ref{fig:FIM_Thru_diff_P} we plot the trade-off between the efficiency of the distance estimation of the user, represented by its Jeffrey's prior\footnote{The estimation error is calculated as the inverse of the Jeffrey's prior.} and the rate coverage probability at a rate threshold of 500 Mbps. Each position in the plot for a given deployment parameter corresponds to a particular $\beta$. Thus for a given power budget, deployment density, and operating beamwidth, the performance of the system is determined by a particular operating characteristic, i.e., a trade-off between the positioning efficiency and data rate performance. For a particular operating characteristic, as we increase $\beta$, we improve the rate coverage probability at the cost of degrading the localization efficiency; whereas, decreasing $\beta$ has the opposite effect. Accordingly, there exists a trade-off between the distance estimation and the data rate performance of the system. In the next subsection, we propose a scheme for selecting $\beta$ based on a given operating beamwidth. 
\begin{figure}
\centering
\includegraphics[width = 8 cm, height = 4cm]{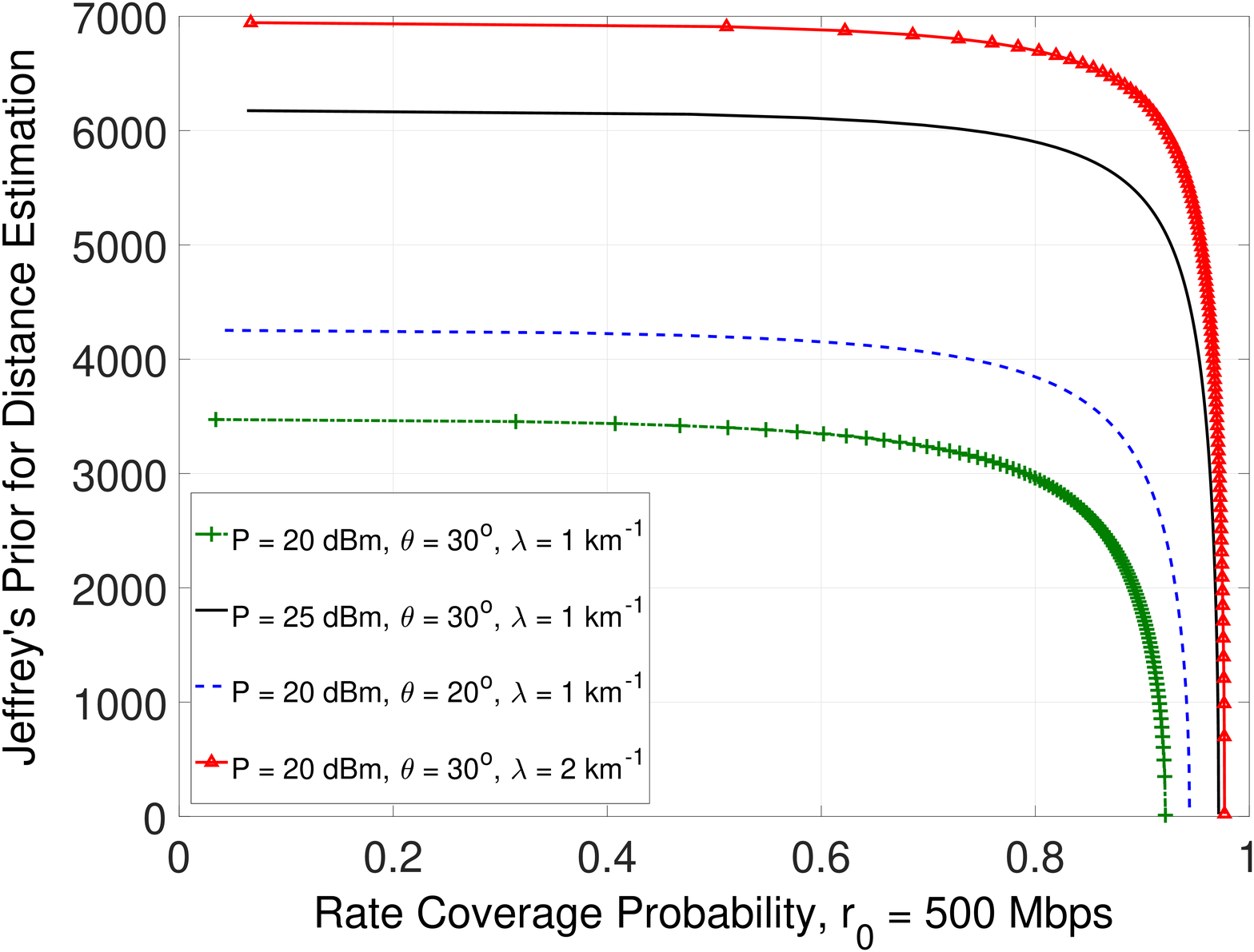}
\caption{Distance estimation error vs physical data rate for different power budget.}
\label{fig:FIM_Thru_diff_P}
\vspace*{-0.5cm}
\end{figure}
\begin{figure}
\centering
\includegraphics[width = 8 cm, height = 4cm]{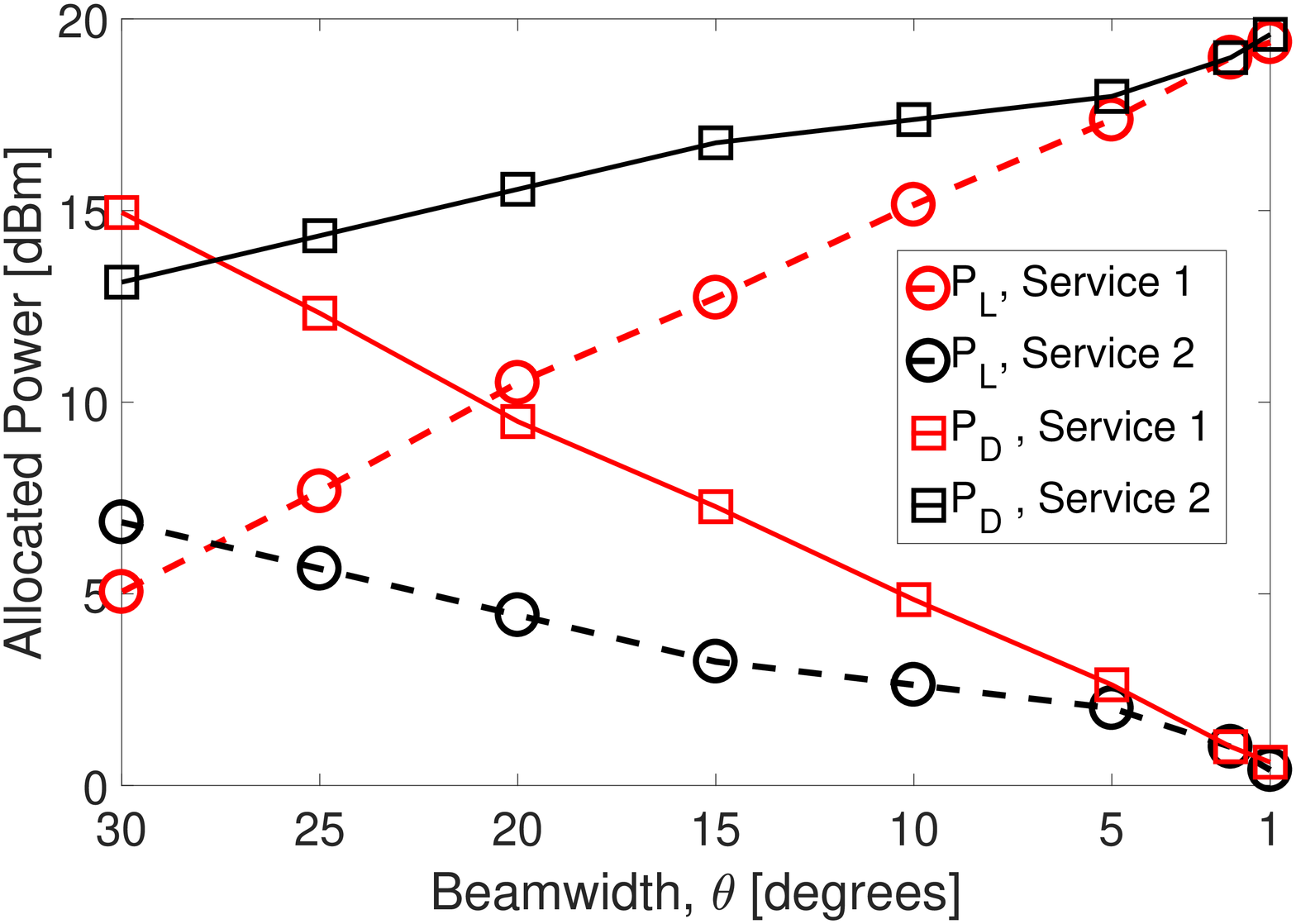}
\caption{Power allocation for the two services for different operating beamwidths.}
\label{fig:final}
\vspace*{-0.5cm}
\end{figure}
\vspace*{-0.4cm}
\subsection{QoS Aware Network Parameter Setting}
We propose the following scheme for setting the network parameters. First, for a given power budget, deployment density and operating beamwidth, the corresponding operating characteristic (i.e., a trade-off curve from Figure~\ref{fig:FIM_Thru_diff_P}) is selected. Next, for the chosen operating characteristic, the minimum $\beta_{min}$ is chosen to satisfy the required outage constraint. Then, for a given positioning error constraint, the maximum value of $\beta$, i.e., $\beta_{max}$ is selected. Finally, the operating $\beta_{min} \leq \beta \leq \beta_{max}$ is selected to address the specific QoS requirements. Accordingly, the misalignment error varies for the chosen $\beta$ and the operating $\theta$.
  
  In what follows, we explain the total power distribution based on the QoS requirements, for a varying degree of misalignment. We assume a network with $\lambda$ = 2 km$^{-1}$ and a BS power budget of $P = 20$ dBm providing two services:
\begin{itemize}
\item Service 1 requires maximum positioning efficiency and a tolerable outage of 10$\%$.
\item Service 2 requires maximum data-rate and a tolerable positioning error of 5e-4 m.
\end{itemize}
We study the power partitioning scheme under different operating beamwidths. In practice, the operating beamwidth may be a system requirement for the first generation mm-wave networks. Intuitively, for a less stringent misalignment requirement, the operating beamwidth can be smaller. This can either be exploited to improve the positioning or enhance the data-rate, as per the required QoS.

For service 1, the operator should set $\beta$ equal to the $\beta_{min}$ corresponding to the $\theta$ that satisfies the misalignment requirement. Then, if the operating $\theta$ can be decreased, more power can be allotted for positioning and the one used for data communication $P\beta_{min}$ is reduced, accordingly. 
On the other hand, the operator should set $\beta$ equal to the $\beta_{max}$ corresponding to the $\theta$ that satisfies the misalignment requirement. Therefore, a thinner beamwidth facilitates larger power allocation for data communication ($P_D$ increases).
The stark difference in the two examples lies in the fact that the advantage of operating with a thinner beamwidth is exploited differently. With decreasing $\theta$, for a positioning service, $P_L$ increases and $P_D$ decreases, whereas the opposite is true for the high data-rate services (see Figure.~\ref{fig:final}).

It is worth mentioning that the inter-dependence of $\beta$ and $\theta$ for controlling the positioning performance and the misalignment error is not trivial. As an example, for a required misalignment constraint or for a required positioning error constraint, there exist non-unique $\left(\theta, \beta\right)$ pairs. 
Furthermore, it may happen that for a given $\theta$ and $P$, no feasible $\beta$ exists that satisfies the positioning and misalignment constraints simultaneously, thereby necessitating a higher BS power budget. This interesting trade-off and the associated optimization problem will be treated in a future work.
\vspace*{-0.2cm}
\section{Conclusion}
In this paper we characterized a mm-wave system deployed to support positioning and broadband services simultaneously. Specifically, we introduced a power-partitioning based mechanism that enables the mm-wave BS to satisfy different localization and data-rate requirements. In this context, we derived dimensioning rules in terms of the density of BSs required to limit outage probability. Then we provided the operator with a beamwidth selection guideline to limit the misalignment probability. Finally, we studied the trade-off between the localization efficiency and the downlink data rate, and consequently, presented a scheme for partitioning the transmit power depending on the service requirements.
\label{sec:Con}
\vspace*{-0.35cm}
	\bibliography{refer.bib}
	\bibliographystyle{IEEEtran}

\end{document}